\documentclass[submission,copyright,creativecommons]{eptcs}
\providecommand{\event}{QPL 2017}

\usepackage[utf8]{inputenc}
\usepackage{scrextend}
\usepackage[english]{babel}
\usepackage{amsmath}
\usepackage{amsfonts}
\usepackage{amsthm}
\usepackage{graphicx}
\usepackage{physics}
\usepackage{mathtools}
\usepackage{xspace}
\usepackage[strings]{underscore}

\theoremstyle{definition}
\newcounter{counter}
\newtheorem{theorem}[counter]{Theorem}
\newtheorem{definition}[counter]{Definition}
\newtheorem{lemma}[counter]{Lemma}
\newtheorem{corollary}[counter]{Corollary}
\newtheorem*{theorem*}{Theorem}
\newtheorem*{lemma*}{Lemma}

\newcommand{\Qstoch}{\text{\textbf{QStoch}}\xspace}
\newcommand{\CPM}{\text{\textbf{CPTP}}\xspace}

\title{Quantum Theory is a Quasi-stochastic Process Theory}
\author{John van de Wetering
\institute{Radboud University\\Nijmegen, Netherlands}
\email{wetering@cs.ru.nl}
}
\def\titlerunning{Quantum Theory is a Quasi-stochastic Process Theory}
\def\authorrunning{J. van de Wetering}

\date{November 2016}

\begin{document}

\maketitle

\begin{abstract}
There is a long history of representing a quantum state using a quasi-probability distribution: a distribution allowing negative values. In this paper we extend such representations to deal with quantum channels. The result is a convex, strongly monoidal, functorial embedding of the category of trace preserving completely positive maps into the category of quasi-stochastic matrices. This establishes quantum theory as a subcategory of quasi-stochastic processes. Such an embedding is induced by a choice of minimal informationally complete POVM's. We show that any two such embeddings are naturally isomorphic. The embedding preserves the dagger structure of the categories if and only if the POVM's are symmetric, giving a new use of SIC-POVM's, objects that are of foundational interest in the QBism community. We also study general convex embeddings of quantum theory and prove a dichotomy that such an embedding is either trivial or faithful.
\end{abstract}

\section*{Introduction}
In Feynman's famous 1981 paper on quantum computation \cite{feynman1982simulating} he writes ``The only difference between a probabilistic classical world and the equations of the quantum world is that somehow or other it appears as if the probabilities would have to go negative''. In this paper we wish to make this statement exact.

Of course, much work has already been done in this regard going all the way back to the Wigner quasi-probability distribution in 1932 \cite{wigner1932quantum}. The Wigner function allows you to associate a probability distribution over a phase space for a quantum particle, with the only caveat that the probability sometimes has to be negative. The negativity appearing in probabilistic representations of quantum systems is something that lies at the heart of quantum theory: Spekkens has shown \cite{spekkens2008negativity} that the necessity of negativity in probabilistic representations is equivalent to the contextuality of quantum theory. It is also a necessity for a quantum speedup as states represented positively by the Wigner function can be efficiently simulated \cite{mari2012positive,pashayan2015estimating,veitch2012negative,veitch2013efficient}. An operational interpretation of negative probabilities is given by Abramsky and Brandenburger \cite{abramsky2014operational}.

The main contribution of this paper will be to represent \emph{all} of (finite-dimensional) quantum theory as a set of quasi-stochastic processes, not just the states. In particular we will use the language of category theory to establish that the category of quantum processes is a subcategory of the category of quasi-stochastic processes. A central object for studying these representations is the \emph{informationally complete POVM} (IC-POVM), this is a measurement that completely characterizes a quantum state. We are particularly interested in IC-POVMs that are minimal as these form a basis of the statespace. A particular type of such a measurement is a symmetric IC-POVM: this is a very special type of POVM that is of particular interest to the QBism community \cite{fuchs2017notwithstanding} as it allows the updating of states to be written in a particularly elegant manner. It also allows quantum states to be written down with a minimal amount of negativity \cite{zhu2016quasiprobability}. Minimal and symmetric POVMs turn out to be essential to preserving respectively the tensor product and adjoint in the quasi-stochastic representations considered in this paper. 

Representing quantum processes by quasi-stochastic matrices is not a new idea. In particular it is used in \cite{appleby2017qplex} to argue the similarity of a unitary process and the Born rule, although they stop short of extending the rule to all quantum channels and of composing them. In \cite{pashayan2015estimating} they also don't go into detail about the compositional nature of these representations either. As far as the author is aware, this paper is the first to consider the compositional structure of quasi-stochastic representations of quantum theory. An approach that comes close is that of the \emph{duotensor} framework of Hardy \cite{hardy2013formalism}, in particular his \emph{hopping metric} is similar to the transition matrices $T$ in this paper, but because Hardy's \emph{fiducial effects} don't have to form a POVM the values in the hopping metric don't form a quasi-probability distribution. In \cite{appleby2017qplex,ferrie2008frame} it was shown that the negativity in the representations can be overcome by modifying the way probabilities are calculated. In this paper this modification takes the form of a category that has a modified composition rule.

We will represent quantum theory using a category \CPM close to that of Selinger's \cite{selinger2007dagger} consisting of trace preserving completely positive maps. This category models the dynamics of finite dimensional quantum systems where throwing away systems is allowed. It excludes classical systems which are, for instance, present in the category of finite dimensional C$^{*}$-algebras. Most of the results carry over to this setting, a small example of which we will give at the end of section \ref{sec:functor}. The category consisting of quasi-stochastic matrices that compose via matrix multiplication will be denoted \Qstoch. The primary contribution of the paper is the following theorem.
\begin{theorem*}
    Any family of minimal informationally complete (IC) POVM's gives rise to a functor $F:\CPM \rightarrow \Qstoch$. This functor is faithful (injective) and strong monoidal (preserving tensor product), it furthermore preserves the convex structure of $\CPM$. The functors arising from two different families of minimal IC-POVM's are naturally isomorphic. The functor preserves the adjoint of unital channels if and only if the POVM's are generalised symmetric informationally complete (SIC).
\end{theorem*}

Note that the fact that any two representations by minimal POVM's are naturally isomorphic in a way seems to answer the question why there doesn't seem to be a preferred way to represent quantum theory by probabilities: the category doesn't `see' this difference. 
This theorem establishes that \CPM can be considered a subcategory of \Qstoch. \CPM is not a \emph{full} subcategory: for a given system not all quasi-probability distributions correspond to valid quantum states. In particular, since no orthogonal IC-POVM exists, the probability distribution that is associated to a quantum state will always be somewhat mixed. This is qualitatively similar to the concept of Spekkens' \emph{epistricted theories} \cite{spekkens2016quasi}. As shown in that paper, many of the characteristic features of quantum theory (no-cloning, teleportation, dense coding, entanglement) arise in classical epistemically restricted theories: theories where states of maximal knowledge are not available. Some features of quantum theory however do \emph{not} occur in classical epistemically restricted theories, primarily Bell / noncontextuality inequality violations and a computational speed-up. As shown by Abramsky and Brandenburger \cite{abramsky2011sheaf}, any non-signalling correlations (including those that maximally violate the Bell inequality) can be represented by a hidden variable model if one allows negative probabilities. Furthermore it was shown that any tomographically local theory has a complexity bound of \textbf{AWPP} \cite{lee2015computation} which was later shown by the same authors to be achieved by a computational model based on a quasi-probabilistic Turing machine \cite{barrett2017localstoch}. These results together suggest that the features of quantum theory that don't occur in classical epistricted theories can be explained by the presence of negative probabilities in quantum theory and that in fact the necessity of this negativity is the `cause' of violating Bell inequalities and achieving a computational speed-up.

We also prove a converse statement to the above theorem. Defining a quasi-stochastic representation of quantum theory to be a convex functor $F:\CPM \rightarrow \Qstoch$ we show the following.
\begin{theorem*}
    To each quasi-stochastic representation we can associate a family of quasi-POVM's (POVM's that don't have to consist of positive components) that determine the representation on the states. Exactly one of the following holds for all quasi-stochastic representations.
    \begin{enumerate}
        \item All the components of the associated quasi-POVM's are multiples of the identity, in which case the representation is trivial.
        \item All the associated quasi-POVM's are informationally complete, in which case the representation is faithful.
    \end{enumerate}
    Furthermore, for nontrivial representations it holds that
    \begin{enumerate}
        \item the representation is strong monoidal if and only if the associated quasi-POVM's are \emph{minimal} informationally complete.
        \item the representation preserves the dagger if and only if the associated quasi-POVM's are \emph{symmetric} informationally complete.
    \end{enumerate}
\end{theorem*}

This gives an interesting new way to look at minimal generalised SIC-POVM's: they are the only POVM's that lead to a quasi-stochastic representation of quantum theory that is both strong monoidal (preserving the tensor product) and that preserves the dagger (the adjoint). The fact that each representation comes from a family of quasi-POVM's mirrors a very similar statements about frames in \cite{ferrie2008frame}.

Any representation is either trivial or faithful and if it is strong monoidal it also has to be minimal. It might even be the case that along the lines of the proof of Theorem \ref{theor:trivialrep} that any nontrivial representation has to be minimal regardless, although we do not show this. If this is in fact the case then it is a strong argument that minimal (quasi-)IC-POVM's are an essential object in quantum theory: they would be the only objects inducing faithful quasi-stochastic representations.

Considering the above results it is a natural question to ask whether any non-signalling theory allows an embedding into the category of quasi-stochastic processes. This turns out to be true for any causal operational-probabilistic theory \cite{chiribella2011informational} which allows coarse-graining and is nondeterministic. This follows easily from the existence of minimal IC measurements in those theories and an adaptation of the proofs in this paper. Such a theory allows a strong monoidal embedding if and only if it allows local tomography. 

 We can also try to start with $\Qstoch$ and identify suitable properties of subcategories that make them isomorphic to $\CPM$. Work has already been done in this direction by Appleby, Fuchs, Schack and others \cite{appleby2017qplex,appleby2011properties,fuchs2013quantum}. They have restricted themselves to looking at state spaces. By considering the entire compositional structure of quantum theory, finding suitable axioms might be easier. In fact, using some simple axioms and the structure of $\Qstoch$ it is possible to derive the modified composition rule $q(j) = \sum_i(\alpha p(i) - \beta)r(j\lvert i)$ used in those papers. This will appear in later work.

We assume some basic familiarity with some concepts of category theory, namely that of a category, functor and natural transformation. All other categorical concepts will be explained when necessary. In section \ref{sec:prelim} we will establish the standard way of associating stochastic matrices and probability distributions to quantum channels and states as done in for instance \cite{ferrie2008frame,ferrie2011quasi}. We will extend this to a functor in section \ref{sec:functor}. The preservation of the tensor product and the adjoint will be studied respectively in sections \ref{sec:tensor} and \ref{sec:adjoint}. Finally, general representations of quantum theory are considered in section \ref{sec:genrep}. The appendix contains some of the longer proofs.

\noindent\textbf{Acknowledgements}:This work is supported by the ERC under the European Union’s Seventh Framework Programme (FP7/2007-2013) / ERC grant n$^\text{o}$ 320571. The author would like to thank the anonymous referees for their valuable comments and for pointing out additional references.

\section{Preliminaries} \label{sec:prelim}
\begin{definition}
Let $A\in M^{n\times m}(\mathbb{R})$ be a $n\times m$ matrix with real entries. It is called \emph{quasi-stochastic} when the values in each column sum up to 1. It is called \emph{stochastic} when it is quasi-stochastic and all the entries are positive. A matrix is called \emph{doubly (quasi-)stochastic} when it is (quasi-)stochastic and its transpose is also (quasi-)stochastic. A doubly quasi-stochastic matrix is necessarily square.
\end{definition}

Note that if a (doubly) quasi-stochastic matrix has an inverse, this inverse will also be (doubly) quasi-stochastic. The inverse of a stochastic matrix is not necessarily stochastic (that is, some negative components might occur).

\begin{definition}
Let $M_n = M^{n\times n}(\mathbb{C})$. We call $\rho \in M_n$ a \emph{state} when $\rho\geq 0$ and $\tr(\rho)=1$. We call $E\in M_n$ an \emph{effect} when $0\leq E\leq 1$. A set of effects $\{E_i\}\subseteq M_n$ is called a \emph{POVM} when $\sum_i E_i = I_n$ where $I_n$ denotes the identity. If the set of $E_i$ are Hermitian and satisfy $\sum_i E_i = I_n$ while no longer being necessarily positive, then we call this set a \emph{quasi-POVM}\footnote{In fact, a quasi-POVM is nothing more than a Hermitian basis where the elements happen to sum up to the identity, we use the term `quasi-POVM' to be consistent about the use of `quasi-' to refer to an object that is usually positive, but in this case not.}. If a (quasi-)POVM spans $M_n$ we call it \emph{Informationally Complete} (IC) and if the elements of a IC-POVM are also linearly independent then it forms a basis for $M_n$ and we call it \emph{minimal} informationally complete. Such a set always has $n^2$ elements.
\end{definition}
Measuring a state $\rho$ using a POVM $\{E_i\}$ leads to probabilities $p(i) = \tr(\rho E_i)$ by the Born rule. Since the $E_i$ are positive and they sum up to $I_n$ these probabilities indeed form a proper probability distribution: $\sum_i p(i) = 1$. 

If the POVM is minimal IC then we can write each $\rho$ uniquely in terms of the $E_i$:
$$\rho = \sum_i \alpha_i \frac{E_i}{\tr(E_i)}.$$
We have chosen to normalise $E_i$ to trace 1 so that the $\alpha_i$ sum up to 1.

Now we can find a relation between the $\alpha$'s and the $p(i)$'s:
$$p(i) = \tr(\rho E_i) = \sum_j \alpha_j \tr(\frac{E_j}{\tr(E_j)}E_i) = \sum_j T(i\lvert j) \alpha_j$$
where $T$ is a matrix defined as
$$T_{ij} = T(i\lvert j) = \tr(\frac{E_j}{\tr(E_j)}E_i).$$

$T$ is a stochastic matrix: $\sum_i T(i\lvert j) = 1$, and it is doubly stochastic if and only if all the $E_i$ have the same trace, which then necessarily has to be $\frac{1}{n}$ because $n = \tr(I_n) = \sum_{i=1}^{n^2} \tr(E_i)$. We will refer to $T$ as a \emph{transition matrix} for $\{E_i\}$.

We have the relation $p = T\alpha$ where we consider $p$ and $\alpha$ as vectors. $T$ has to be invertible because the $E_i$ form a basis, so $\alpha = T^{-1} p$, which means we can write $\rho$ in terms of its probabilities over $E_i$:
$$\rho = \sum_i (T^{-1}p)_i \frac{E_i}{\tr(E_i)}.$$
Note that while $T$ is stochastic, $T^{-1}$ will in every case contain some negative elements and so will be just quasi-stochastic \cite{ferrie2008frame}. This is not too surprising because if we had a POVM that has a $T^{-1}$ that is stochastic we would have succeeded in finding a non-contextual hidden variable model for quantum theory.

Now suppose we have for $M_m$ a minimal IC-POVM $\{E_i^\prime\}$. And consider a completely positive trace preserving (CPTP) map $\Phi: M_n\rightarrow M_m$. Let $\sigma = \Phi(\rho)$. We wish to know the probability distribution of $\sigma$ over $E_i^\prime$ in terms of the probability distribution of $\rho$ over $E_i$. Using the expansion of $\rho$ in terms of $p$:

$$q(i) = \tr(\sigma E_i^\prime) = \tr(\Phi(\rho)E_i^\prime) = \sum_j (T^{-1}p)_j \tr(\Phi\left(\frac{E_j}{\tr(E_j)}\right)E_i^\prime) = \sum_j r(i\lvert j) (T^{-1}p)_j$$
where we have introduced a matrix $r$ defined as
$$r_{ij} = r(i\lvert j) =  \tr(\Phi\left(\frac{E_j}{\tr(E_j)}\right)E_i^\prime).$$
The matrix $r$ is again stochastic because $\Phi$ is trace preserving. The matrix is doubly stochastic if and only if the $E_i$ and $E_j^\prime$ all have the same trace and $\Phi$ is unital.

We see that we have now translated the equation $\sigma = \Phi(\rho)$ into the equation $q = rT^{-1}p$. We can also translate composition of maps. Let $\Psi: M_m\rightarrow M_k$ be CPTP and $\{E_i^{\prime\prime}\}$ be a minimal IC-POVM on $M_k$. Set
$$
T^\prime(i\lvert j) = \tr(\frac{E_j^\prime}{\tr(E_j^\prime)}E_i^\prime), \qquad
s(i\lvert j) = \tr(\Psi\left(\frac{E_j^\prime}{\tr(E_j^\prime)}\right)E_i^{\prime\prime}).
$$

Then setting $\sigma^\prime = (\Psi\circ \Phi)(\rho) = \Psi(\Phi(\rho)) = \Psi(\sigma)$ and letting $q^\prime$ be the probability distribution associated to $\sigma^\prime$  we can derive in the same way as before that
$$q^\prime = s(T^\prime)^{-1} q = s(T^\prime)^{-1} r T^{-1} p.$$
Since the map $\Psi\circ\Phi$ is completely determined by what it does on states this also completely determines the matrix that we should associate with it, namely: $s(T^\prime)^{-1} r$.

Note that the constructions in this section work equally well with quasi-POVMs, replacing probabilities with quasi-probabilities and stochastic matrices with quasi-stochastic matrices.

\section{Functorial embedding from POVM's} \label{sec:functor}
We've seen how to translate composition of quantum channels into composition of quasi-stochastic matrices. To make this transition formal we will show that this induces a functor between the relevant categories.

\begin{definition}
\CPM is the category which has as objects the natural numbers $n>0$, and its morphisms $\Phi: n\rightarrow m$ are CPTP maps $\Phi: M_n\rightarrow M_m$. Composition is the usual composition of linear maps. \emph{States} are maps $\tilde{\rho}: 1 \rightarrow n$.
\end{definition}
Note that the definition of states corresponds to the definition of state we have used above, because if $\tilde{\rho}: M_1 = \mathbb{C} \rightarrow M_n$ then it is completely defined by its action on $1$, and we see that $\tilde{\rho}(1) = \rho\in M_n$ is positive and trace 1. We will most of the time simply use the actual states and not the morphism using the abuse of notation $\Phi(\rho) := \Phi \circ \tilde{\rho}$.

\begin{definition}
\textbf{QStoch} is the category which has as objects the natural numbers $n>0$ and its morphisms $A: n\rightarrow m$ are quasi-stochastic matrices $A\in M^{n\times m}(\mathbb{R})$ and composition works by regular matrix multiplication. States $p: 1 \rightarrow n$ correspond to quasi-probability distributions.
\end{definition}

We saw above that when we associate stochastic matrices $s$ and $r$ to $\Psi$ and $\Phi$, the composition $\Psi\circ \Phi$ has the matrix $sT^{-1} r$ associated to it, where $T$ is determined by the choice of POVM on the intermediate space. We can capture this in a category by modifying the notion of composition: $s*r := sT^{-1}r$.

\begin{definition}
Let $T = (T^{(i)})_{i=1}^\infty$ be a family of invertible quasi-stochastic matrices where $T^{(i)}$ is a $i\times i$ square matrix. Define \textbf{QStoch}$_T$ to be the category with the same objects and morphisms as \textbf{QStoch} but with composition of $r: n\rightarrow m$ and $s: m\rightarrow k$ defined as $s*r = s(T^{(m)})^{-1}r$.
\end{definition}
The associativity of this composition follows from the associativity of matrix multiplication and the new identity morphism for $n$ is now $T^{(n)}$ so that $\Qstoch_T$ is indeed a category for any choice of invertible quasi-stochastic matrices.

\begin{theorem}
\label{theor:embed}
Fix for every finite dimension $n>0$ a minimal IC-POVM $\{E_i^{(n)}\}$ and let $T = (T^{(i)})$ be the family of matrices where $T^{(i)} = T_n$ is the transition matrix for $\{E_j^{(n)}\}$ as defined above when $i=n^2$ and otherwise let it be the identity matrix. Then $Q:\CPM \rightarrow \Qstoch_{T}$ defined by $Q(n) = n^2$ and 
$$ Q(\Phi: n \rightarrow m)_{ij} = \tr(\Phi\left(\frac{E_j^{(n)}}{\tr(E_j^{(n)})}\right)E_i^{(m)})$$
is a faithful functor.
\end{theorem}
\begin{proof}
We have defined the composition in the category $\Qstoch_T$ precisely so that $Q$ preserves composition: $Q(\Psi\circ \Phi) = Q(\Psi)*Q(\Phi)$, and it is easy enough to realise that $Q(id_n) = T_n$ which acts as the identity in $\Qstoch_T$, so that $Q$ is indeed a functor.

Note that $M_1 = \mathbb{C}$ has a unique choice for a minimal IC-POVM, namely $\{1\}$. We then see that $Q(\tilde{\rho}: 1\rightarrow n)_{i1} = \tr(\tilde{\rho}(1)E_i^{(n)}) = \tr(\rho E_i^{(n)})$ which is just the Born rule as expected.

Since the $E_i$ are informationally complete, all the states are sent to different probability distributions. Now suppose $Q(\Phi)=Q(\Psi)$. Then we also get $Q(\Phi\circ \rho) = Q(\Phi)*Q(\rho) = Q(\Psi)*Q(\rho) = Q(\Psi\circ \rho)$ so that we must have $\Phi\circ \rho = \Psi \circ \rho$ for all states $\rho$. Since a map is completely defined by its action on states we must then have $\Phi = \Psi$, so that $Q$ is indeed faithful.
\end{proof}

Note that this functor maps all states and maps to distributions and matrices with positive entries. All the negativity is hidden in the modified composition. This is analogous to the \emph{modified probability calculus} of \cite{ferrie2008frame} and the \emph{urgleichung} of \cite{appleby2017qplex}. We can bring this negativity more to the forefront using the following result.

\begin{theorem}
Let $T = (T^{i})_{i=1}^\infty$ be a family of invertible quasi-stochastic matrices of size $i\times i$.  The functor $F_T:\Qstoch_T \rightarrow \Qstoch$ defined by $F_T(n) = n$ and $F_T(A: n \rightarrow m) = A(T^{(n)})^{-1}$ is an isomorphism of categories.
\end{theorem}
\begin{proof}
    The identity morphism of $n$ in $\Qstoch_T$ is $T^{(n)}$, so that $F_T(id_n) = F_T(T^{(n)}) = T^{(n)}(T^{(n)})^{-1} = I_n$. That it preserves composition follows easily from the definition of composition in $\Qstoch_T$, so it is indeed a functor. That it is an isomorphism of categories follows because it has an inverse functor $F_T^{-1}(A) = AT^{(n)}$.  
\end{proof}

Now we see that the composition $F_T\circ Q$ gives for any family of minimal IC-POVM's an embedding of quantum theory into the category of quasi-stochastic maps. States are all mapped to proper probability distributions (no negative components), while effects do contain negative components. Instead of $F_T(A) = A(T^{(n)})^{-1}$ we could have also defined the isomorphism $F_T^\prime(A) = (T^{(m)})^{-1}A$. In that case effects would be mapped to proper probabilistic effects and states would instead contain negative components.

There are a lot of ways to represent quantum theory as a quasi-stochastic theory, but one of the problems is that it is hard to find a reason to prefer one over the other. Using the functorial embedding we can make it clear why there doesn't seem to be a preferred one:

\begin{theorem}\label{theor:nateq}
Let $\{E_i^{(n)}\}$ and $\{F_i^{(n)}\}$ be families of minimal IC-POVM's, and let the $T^1$ and $T^2$ be the respective families of matrices as defined in Theorem \ref{theor:embed}, and let $Q_i:\CPM \rightarrow \Qstoch_{T_i}$ be the respective functors. Then there exists a natural isomorphism $\eta: F_{T^1}\circ Q_1 \Rightarrow F_{T^2}\circ Q_2$.
\end{theorem}
\begin{proof}
See appendix \ref{app:nateq}.
\end{proof}

Any minimal embedding is naturally isomorphic to another, which means that as far as the categories are concerned there really is only one embedding. Note that the category \Qstoch doesn't ``see'' negativity, so these embeddings can still be very different in terms of which maps are represented with negative components.

We've now shown that all quantum channels and all states can be represented in a coherent manner in terms of quasi-probabilities, but this is not really \emph{all} of quantum theory. We should also consider as done in \cite{appleby2017qplex,ferrie2008frame} the probabilities arising from measuring a state using a POVM. So lets consider a measurement with $K$ different outcomes where the probabilities for a state $\rho$ are given as $P(k\lvert \rho) = \tr(A_k\rho)$ where $\{A_k\}_{k=1}^K$ denote a POVM on $M_n$. Expanding $\rho$ in terms of its stochastic representation: $P(k\lvert \rho) = \sum_i (T_n^{-1}Q(\rho))_i\tr(A_k E_i/\tr(E_i))$. It is helpful to consider the POVM $\{A_k\}$ as a map $\hat{A}:M_n \rightarrow \mathbb{C}^K$ in the category of C$^*$-algebras where $\hat{A}(B)_k = \tr(A_kB)$. Abusing notation a bit we can then write
$$
    Q(\hat{A})(k\lvert i) = \langle e_k,\hat{A}\left(\frac{E_i}{\tr(E_i)}\right)\rangle = \tr(A_k \frac{E_i}{\tr(E_i)})
$$

where we let $Q$ map $\mathbb{C}^K$ to $k$ and we take the standard basis $\{e_k\}$ as `the POVM' for $\mathbb{C}^K$. This makes sense when we view $\mathbb{C}^K$ as the diagonal matrices in $M_K$ in which case the standard basis components $e_k$ correspond to the diagonal projections to the $k$th component. It is easy to check that $Q(\hat{A})$ is indeed a stochastic matrix and we see that $P(k\lvert \rho) = \sum_i Q(\hat{A})(k\lvert i)(T_n^{-1}Q(\rho))_i$ so that
$$
P(\cdot \lvert \rho) = Q(\hat{A})T_n^{-1}Q(\rho) = Q(\hat{A})*Q(\rho).
$$
This looks exactly the same as applying a CPTP map to a state. In this view measuring a POVM corresponds to a `quantum-classical' channel that takes as input a quantum state and outputs a classical state (a probability distribution).

\section{Preserving the tensor product} \label{sec:tensor}
An important part of quantum theory is the possibility of parallel composition: the tensor product. This can be captured by the fact that the category of quantum processes \CPM is a (strict) \emph{monoidal category}\footnote{We will work exclusively with strict monoidal categories in this paper, so we will ignore the coherence isomorphisms that usually appear.}:

\begin{definition}
    A category $\mathbb{A}$ is called \emph{strict monoidal} when it has a bifunctor $\otimes: \mathbb{A}\times \mathbb{A} \rightarrow \mathbb{A}$, and an identity object $I$, such that for all objects $A,B,C \in \mathbb{A}$: $I\otimes A = A\otimes I = A$ and $A\otimes (B\otimes C) = (A\otimes B)\otimes C$. 
\end{definition}
The functoriality condition boils down to $id_A\otimes id_B = id_{A\otimes B}$ and $(f_1\otimes f_2)\circ (g_1\otimes g_2) = (f_1\circ g_1)\otimes (f_2\circ g_2)$ for all compatible morphisms.

It is easy to verify that the linear algebraic tensor product turns \CPM and \Qstoch into monoidal categories, where on objects it acts as $n\otimes m = nm$. 

The relevant notion of a morphism between monoidal categories is that of a strong monoidal functor.
\begin{definition}
    A functor between two strict monoidal categories $F: \mathbb{A}\rightarrow \mathbb{B}$ is called \emph{strong monoidal} when the two functors $F_1(A,B) = F(A)\otimes F(B)$ and $F_2(A,B) = F(A\otimes B)$ are naturally isomorphic via $\alpha: F_1\Rightarrow F_2$ such that the components $\alpha_{A,B}: F(A)\otimes F(B) \rightarrow F(A\otimes B)$ satisfy the coherence condition $\alpha_{A\otimes B, C}\circ (\alpha_{A,B}\otimes id_C) = \alpha_{A,B\otimes C}\circ(id_A \otimes\alpha_{B,C})$. The naturality condition means that $\alpha_{B_1,B_2}\circ (F(f_1)\otimes F(f_2)) = F(f_1\otimes f_2)\circ \alpha_{A_1,A_2}$ for all morphisms $f_i:A_i\rightarrow B_i$.

    The functor is called \emph{strict monoidal} when all the $\alpha$ are identities and it is called \emph{lax monoidal} when the $\alpha$ are not necessarily isomorphisms.
\end{definition}

There are multiple choices for a tensor product in $\Qstoch_T$. We will choose the tensor product such that the functor $F^\prime:\Qstoch_T \rightarrow \Qstoch$ defined by $F^\prime(A:n\rightarrow m) = T_m^{-1} A$ is strict monoidal. Denote the tensor product in $\Qstoch_T$ by $\otimes^\prime$, then we should have $F^\prime(A\otimes^\prime B) = F^\prime(A)\otimes F^\prime(B)$. Writing this out we get 
$$
    T_{m_1m_2}^{-1} (A\otimes^\prime B) = (T_{m_1}^{-1}\otimes T_{m_2}^{-1})(A\otimes B)
$$ 
so that $\otimes^\prime$ should be defined as $A\otimes^\prime B := T_{m_1m_2}(T_{m_1}^{-1}\otimes T_{m_2}^{-1})(A\otimes B).$
It is instructive to check that this indeed turns $\Qstoch_T$ into a monoidal category:
\begin{align*}
    (A_1\otimes^\prime A_2)*(B_1\otimes^\prime B_2) &= T_{k_1k_2}(T_{k_1}^{-1}\otimes T_{k_2}^{-1})(A_1\otimes A_2) T_{m_1m_2}^{-1}T_{m_1m_2}(T_{m_1}^{-1}\otimes T_{m_2}^{-1})(B_1\otimes B_2) \\
    &= T_{k_1k_2}(T_{k_1}^{-1}\otimes T_{k_2}^{-1}) (A_1T_{m_1}^{-1}B_1)\otimes (A_2T_{m_2}^{-1}B_2) \\
    &= (A_1*B_1)\otimes^\prime (A_2*B_2)
\end{align*}

We can now check that the functor $F: \Qstoch_T\rightarrow \Qstoch$ is strong monoidal. We have $F(A\otimes^\prime B) = T_{m_1m_2}(T_{m_1}^{-1}\otimes T_{m_2}^{-1})(A\otimes B) T_{n_1n_2}^{-1}$ and $F(A)\otimes F(B) = (A\otimes B)(T_{n_1}^{-1}\otimes T_{n_2}^{-1})$. This can be rewritten to
$$
    T_{m_1m_2}(T_{m_1}^{-1}\otimes T_{m_2}^{-1}) (F(A)\otimes F(B)) = F(A\otimes^\prime B)T_{n_1n_2}(T_{n_1}^{-1}\otimes T_{n_2}^{-1})
$$
which means our coherence isomorphisms are $\alpha_{n_1,n_2} = T_{m_1m_2}(T_{m_1}^{-1}\otimes T_{m_2}^{-1})$. It then follows by straightforward matrix multiplication that these satisfy the coherence conditions.

Note that we could have chosen the tensor product in $\Qstoch_T$ such that $F$ would be strict monoidal and $F^\prime$ would be strong monoidal. The reason we have chosen this tensor product is that it makes the following easier.

\begin{theorem}\label{theor:monoid}
    The functor $Q: \CPM\rightarrow \Qstoch_T$ as defined in Theorem \ref{theor:embed} for a family of minimal IC-POVM's is strong monoidal with coherence isomorphisms $\alpha_{n_1,n_2}=S_{n_1,n_2}$, where
    $$S_{n_1,n_2}(j\lvert i_1i_2) = \tr(\frac{E_{i_1}^{n_1}}{\tr(E_{i_1}^{n_1})}\otimes \frac{E_{i_2}^{n_2}}{\tr(E_{i_2}^{n_2})} E_j^{n_1n_2})
    $$
\end{theorem}
\begin{proof}
See appendix \ref{app:monoid}.
\end{proof}

The composition of two strongly monoidal functors is again strongly monoidal, so $F\circ Q$ gives a strong monoidal functor of \CPM into \Qstoch. The functor is \emph{strong} monoidal precisely because the $S_{n_1,n_2}$ are invertible and therefore are isomorphisms. These $S$ are invertible because the POVM's are minimal. If they weren't minimal we wouldn't get a strong monoidal functor, but at most a \emph{lax} monoidal functor.

\section{Preserving the adjoint} \label{sec:adjoint}
Another structure that exists on quantum channels is the \emph{adjoint}. If we have a CP map $\Phi: M_n\rightarrow M_m$, its adjoint is the dual of the map with respect to the Hilbert-Schmidt inner product. That is, a map $\Phi^\dagger: M_m \rightarrow M_n$ (note that the direction of the map has changed) such that
$$
    \langle \Phi(A), B \rangle_{HS} = \Tr(\Phi(A)B^\dagger) = \Tr(A(\Phi^\dagger(B))^\dagger) = \langle A,\Phi^\dagger(B) \rangle_{HS}.
$$
If $\Phi(A) = UAU^\dagger$, then $\Phi^\dagger(A) = U^\dagger A U$, e.g., for unitary conjugation the adjoint is simply the Hermitian adjoint of the unitary. Hamiltonian evolution in time is represented by the unitary $\exp(itH)$ where the Hermitian adjoint is $\exp(-itH)$. The adjoint can therefore be interpreted as reversing the time direction. It is also how one transfers between the Schrödinger and Heisenberg picture of quantum mechanics.
The adjoint has the property that it distributes over composition, while changing the order of the morphisms: $(\Psi\circ \Phi)^\dagger = \Phi^\dagger \circ \Psi^\dagger$. This behaviour can be defined in a more general way:
\begin{definition}
    A \emph{dagger category} $\mathbb{A}$ is a category which has a contravariant functor $\dagger: \mathbb{A}\rightarrow \mathbb{A}$ that is the identity on objects and is its own inverse. Concretely this means that $id_A^\dagger = id_A$ and for all morphisms $(f^\dagger)^\dagger = f$, while composition satisfies $(g\circ f)^\dagger = f^\dagger \circ g^\dagger$. We will call a category a \emph{partial dagger category}, when the dagger is only defined for a subset of the morphisms.
    A \emph{dagger functor} is a functor between (partial) dagger categories $F: \mathbb{A}\rightarrow \mathbb{B}$ that preserves the dagger (when it exists) in the following way: $F\circ \dagger_\mathbb{A} = \dagger_\mathbb{B}\circ F$, e.g. $F(f^{\dagger_\mathbb{A}}) = (F(f))^{\dagger_\mathbb{B}}$.
\end{definition}

The category \CPM is not a dagger category as the adjoint of a trace preserving map $\Phi$ will only be trace preserving when $\Phi$ is unital. We could consider the category of trace preserving unital maps, but this is a bit restrictive in that it doesn't have morphisms between different objects. Therefore we will stick with \CPM as a partial dagger category and only define the dual of a map when it is unital.

\Qstoch is also only a partial dagger category with the transpose taking the role of the dagger. The transpose of a quasi-stochastic matrix is only quasi-stochastic when it is doubly quasi-stochastic (this is the exact analogue of trace preserving unital maps). We could again `fix' this by only considering doubly quasi-stochastic maps, but instead we will only define the dagger when a map is doubly quasi-stochastic.

$\Qstoch_T$ doesn't always have a partial dagger structure. The obvious choice is the transpose: $(A*B)^T = (AT^{-1}B)^T = B^T (T^T)^{-1}A^T$. We need $(A*B)^T = B^T * A^T$ so that equality here only holds when $T = T^T$. Looking at the definition of the transition matrix $T$ this is the case when all the components of the POVM have equal trace. In that case we define the partial dagger on $\Qstoch_T$ to be the transpose as defined on doubly quasi-stochastic matrices.

\begin{lemma}
    Let $\{E_i^{(n)}\}$ be a family of minimal IC-POVM's where all the components in a single POVM have equal trace, Then the functor $Q$ as in Theorem \ref{theor:embed} is a dagger functor.
\end{lemma}
\begin{proof}
    The dagger in \CPM is only defined when $\Phi: M_n\rightarrow M_n$ is CPTP unital, which necessitates that the input and output dimension are equal. Let $\{E_i\}$ be the POVM for $n$ (dropping the superscript). We have $\tr(E_i) = \frac{1}{n}$ for all $i$, because the POVM is minimal and all the traces are equal. So now
    $$
        Q(\Phi^\dagger)_{ij} = \tr(\Phi^\dagger\left(\frac{E_j}{\tr(E_j)}\right)E_i) = \frac{1}{n}\tr(E_i\Phi^\dagger(E_j)) = \frac{1}{n}\tr(E_i(\Phi^\dagger(E_j))^\dagger) = \frac{1}{n}\tr(\Phi(E_i)E_j^\dagger) = Q(\Phi)_{ji}
    $$
    where $E_i^\dagger = E_i$, because they are by definition Hermitian. Indeed $Q(\Phi^\dagger) = Q(\Phi)^T$ as required.
\end{proof}

Somewhat surprisingly this won't in general extend to a dagger functor from \CPM to \Qstoch by appending $Q$ with $F:\Qstoch_T \rightarrow \Qstoch$. Let $A$ be a $n\times n$ doubly quasi-stochastic matrix and assume that $T$ is a symmetric matrix. Recall that $F(A) = AT^{-1}_n$ which gives $F(A^T) = A^T T^{-1}_n$ while $F(A)^T = (AT^{-1}_n)^T = T^{-1}_n A^T$. Now, of course in general $A^T T^{-1}_n \neq T^{-1}_n A^T$ so that this functor is \emph{not} a dagger functor. In fact, it will only be a dagger functor if $T$ commutes with all quasi-stochastic matrices.

\begin{lemma}
    Let $J_n$ denote the $n\times n$ matrix with every component equal to 1: $(J_n)_{ij} = 1$. Suppose the $n\times n$ doubly quasi-stochastic matrix $T$ commutes with all $n\times n$ doubly quasi-stochastic matrices, then there are real constants $\alpha,\beta$ such that $T = \alpha I_n + \beta J_n$ where $\alpha + n\beta = 1$.
\end{lemma}
\begin{proof}
    All matrices will be $n\times n$ square matrices. We note that a matrix $A$ is quasi-stochastic if and only if $JA = J$. It is furthermore doubly quasi-stochastic if and only if $AJ = JA = J$. $J$ is obviously a rank 1 matrix, which has a single non-zero eigenvector: \textbf{1} $ = (1,\ldots,1)$, such that $J$\textbf{1} $ = n$\textbf{1}. We therefore see that $P = \frac{1}{n}J$ is a 1-dimensional projection, and that the linear subspace spanned by the doubly quasi-stochastic matrices is exactly $\{A ; [A,P] = 0 \}$. This commutation essentially fixes one of the eigenvalues of the matrices, but apart from that doesn't require anything extra. This space is therefore isomorphic to $\mathbb{R}\times M_{n-1}$. Now any doubly quasi-stochastic matrix $A$ can be written as $PAP + (I_n-P)A(I_n-P)$. Therefore for $T$ to commute with a doubly quasi-stochastic matrix $A$ we need $(I_n-P)T(I_n-P)$ to commute with $(I_n-P)A(I_n-P)$ (since it automatically commutes on the other part). But $(I_n-P)A(I_n-P)$ is just an arbitrary $(n-1)\times (n-1)$ matrix. We know that the only matrices in the space $M_k$ that commute with all the matrices are multiples of the identity. We therefore have $(I_n-P)T(I_n-P) = \alpha(I_n-P)$ (noting that $I_n-P$ acts as the identity on this space). And so in fact 
    $$T = \alpha(I_n-P) + \beta P = \alpha I_n + (\beta-\alpha) P = \alpha I_n + \frac{1}{n}(\beta-\alpha)J$$
    which with a proper relabelling becomes $T = \alpha I_n + \beta J$. $T$ is automatically symmetric, and it is quasi-stochastic exactly when $\alpha + n\beta = 1$.
\end{proof}

So for $F$ to be a dagger functor, the transition matrices have to be particularly simple: $T_n = \alpha_n I_n + \beta_n J_n$. Let $\{E_i\}$ be a minimal IC-POVM for $M_n$. Suppose its transition matrix has that form, then
$$
    T_{ij} = \tr(\frac{E_j}{\tr(E_j)}E_i) = \alpha \delta_{ij} + \beta.
$$
We note that the right-hand side is symmetric, so that the left-hand side has to be as well, which means that $\tr(E_j) = \frac{1}{n}$ by minimality. This POVM has a very specific symmetry property.
\begin{definition}
We say that a POVM $\{E_i\}$ is \emph{generalised symmetric informationally complete} (generalised SIC) for $M_n$ when it is minimal informationally complete and
$$
    \tr(E_iE_j) = \alpha\delta_{ij} + \beta
$$
for some constants $\alpha$ and $\beta$. It is called SIC (not generalised) when all the components are rank-1: $E_i = \frac{1}{n}\Pi_i$ for some projections $\Pi_i$. For such a SIC the constants need to be $\alpha = \frac{1}{n^2}$ and $\beta = \frac{n-1}{n^4}$.
\end{definition}
The existence and behaviour of such POVM's has been subject of intense study especially in the case of the rank 1 SIC's where the existence in arbitrary dimension is still not proven\cite{appleby2017qplex,scott2010symmetric}. The generalised SIC's have been classified\cite{gour2014construction} and there exist many of them for each dimension.

With the previous lemma and our observations about the functor preserving the dagger we get the following theorem.
\begin{theorem}\label{theor:dagger}
    Let $Q$ be the functor from Theorem \ref{theor:embed}. $F\circ Q: \CPM \rightarrow \Qstoch$ preserves the dagger if and only if all POVM's are generalised SIC.
\end{theorem}
\begin{proof}
    $Q$ always preserves the dagger, so $F\circ Q$ only preserves the dagger when $F$ does. We've already seen that $F$ preserves the dagger if and only if all the transition matrices have the special symmetric form corresponding to generalised SIC-POVMs.
\end{proof}

Recalling that the adjoint of a map can be interpreted as its time reversal, this gives an interpretation of SIC-POVMs as being the only POVMs `preserving the arrow of time' in the sense that the image of the time-reversal of a map is the time-reversal of its image. 

\section{General functorial embeddings}\label{sec:genrep}
So far we have only considered minimal IC-POVM's: A POVM the components of which form a basis for the underlying matrix space. However we can also consider POVM's consisting of more components. In that case the components won't be linearly independent so that a state can be represented in multiple ways as a combination of the POVM elements:
$$ \rho = \sum_i \alpha_i \frac{E_i}{\tr(E_i)} = \sum_i \alpha^\prime_i \frac{E_i}{\tr(E_i)}.$$
This means that 
$$
\tr(\rho E_i) = \sum_j \alpha_j \tr(\frac{E_j}{\tr(E_j)}E_i) = \sum_j T(i\lvert j)\alpha_j = \sum_j T(i\lvert j)\alpha^\prime_j
$$
or as vectors: $p = T\alpha = T\alpha^\prime$. This is the case because $T$ is no longer a full rank matrix so that it has a nontrivial kernel. Now we can consider the generalised inverse of $T$ defined as having the same kernel as $T$ but otherwise acting as an inverse which gives a canonical choice of $\alpha$: $\alpha = T^{-1}p$. 

We can also still transfer composition: $\Phi(\rho) \mapsto Q(\Phi)T^{-1}Q(\rho)$. This is well-defined in the sense that ker$(T)\subset $ ker$(Q(\Phi))$, but we can no longer define a category $\Qstoch_T$ as there is no longer an identity morphism: $A*T\neq A$, because $T^{-1}T \neq I_n$. The composite functor $F\circ Q$ where we map $\Phi \mapsto Q(\Phi)T^{-1}$ also doesn't work any more as we still don't map the identity to the identity. However, we could consider that there are some smart choices which do produce a valid functor. For instance there is a priori no reason that we can't let $id\mapsto I$. To preserve the convex structure we should then also change where the other maps are sent to. 

Suppose we have such a functor, how much of what we previously constructed still goes through? Looking at Theorem \ref{theor:nateq} we see that we no longer get a natural isomorphism as the matrices involved are no longer invertible. We also see that the proof of Theorem \ref{theor:monoid} explicitly requires the minimality of the representation since otherwise the morphisms involved wouldn't be isomorphisms. The most a nonminimal representation can therefore be is lax monoidal. There are no obvious barriers to a nonminimal representation preserving the adjoint.

So far we have studied embeddings of \CPM into \Qstoch coming from POVMs. Let's consider a bit more general view.

\begin{definition}
    We call $F: \CPM \rightarrow \Qstoch$ a \emph{quasi-stochastic representation} of \CPM when $F$ is a functor that preserves the convex structure of the categories: $F(t\Phi_1 + (1-t)\Phi_2) = tF(\Phi_1) + (1-t)F(\Phi_2)$ and maps $1$ to $1$ (so that states are sent to states). We call it a \emph{positive representation} when it sends all states to nonnegative distributions. A representation is called \emph{strong monoidal / faithful / dagger preserving} when the functor is.
    We call a representation \emph{trivial} when for all $n\in \mathbb{N}$ $F(\rho)=F(\sigma)$ for all states $\rho,\sigma\in M_n$.
\end{definition}

\begin{lemma}
    Let $F:\CPM \rightarrow \Qstoch$ be a quasi-stochastic representation. Then we can find for every $n\in\mathbb{N}$ a quasi-POVM $\{E_i^{(n)}\}$ such that for a state $\rho \in M_n$ we have $F(\rho)_i = \tr(\rho E_i^{(n)})$.
\end{lemma}
\begin{proof}
Let $n\in \mathbb{N}$. We see that $F$ restricts to a convex map $f^{(n)}:$ DO$(n) \rightarrow $ QProb$(F(n))$ where DO$(n)\subset M_n$ represents the set of density operators and QProb$(k_n)$ is the set of quasi-probability distributions on $F(n)$ points. We can split this map up into its components $f^{(n)}_i:$ DO$(n) \rightarrow \mathbb{R}$. These functions are again convex and using standard methods we can extend these first to linear maps on Hermitian matrices and then to linear maps to all matrices: $f^{(n)}_i: M_n\rightarrow \mathbb{C}$. This is exactly an element of the dual space of $M_n$ and since this is finite dimensional it is isomorphic to $M_n$ which means there exists a $E_i^{(n)}\in M_n$ such that $f^{(n)}(A) = \langle A, E_i^{(n)} \rangle_{HS} = \tr(A (E_i^{(n)})^\dagger)$. Because this function is real valued on the Hermitian matrices, $E_i^{(n)}$ has to be Hermitian and because we have $1 = \sum_i f_i^{(n)}(\rho) = \tr(\rho \sum_i E_i^{(n)})$ we get $\sum_i E_i^{(n)} = I_n$ proving that $\{E_i^{(n)}\}$ is indeed a quasi-POVM.
\end{proof}
We will refer to these quasi-POVM's as the representation's associated quasi-POVM's. It is easy to see that the representation is positive if and only if the associated quasi-POVM's are true POVM's (consisting of only positive components).

\textbf{Note:} If we start out with a family of POVM's and construct the regular functor $Q: \CPM \rightarrow \Qstoch_T$ and then append this with the functor $F^\prime: \Qstoch_T\rightarrow \Qstoch$ given by $F^\prime(A) = T^{-1}A$ then $F^\prime\circ Q$ will not be a positive representation, while $F \circ Q$ where $F(A) = AT^{-1}$ will be. $F^\prime\circ Q$ and $F\circ Q$ are naturally isomorphic, so positivity of the representation is not a `categorical' property in the sense that it isn't preserved by natural isomorphism.

It is not hard to see that a representation is faithful if and only if the associated quasi-POVM's are informationally complete. On the other hand, suppose a representation is trivial, then for a element of the quasi-POVM $E\in M_n$ we must have $\tr(\rho E) = \tr(\sigma E)$ for all states $\rho,\sigma\in M_n$. Calling this value $\alpha$ and realising that $\alpha = \tr(\rho \alpha I_n)$ we see that $\tr(\rho (E-\alpha I_n))=0$ for all states $\rho$ so that $E= \alpha I_n$. Perhaps somewhat surprisingly these are the only options for the associated POVM's:

\begin{theorem} \label{theor:trivialrep}
    Let $F: \CPM \rightarrow \Qstoch$ be a quasi-stochastic representation of $\CPM$. This representation is either faithful or trivial.
\end{theorem}
\begin{proof}
    See appendix \ref{app:trivialrep}.
\end{proof}

This means that we immediately get the following corollary using Theorems \ref{theor:monoid} and \ref{theor:dagger}.

\begin{corollary}
    Let $F: \CPM \rightarrow \Qstoch$ be a positive nontrivial representation. 
    \begin{itemize}
        \item $F$ is strong monoidal if and only if the associated POVM's are minimal IC.
        \item $F$ preserves the dagger if and only if the associated POVM's are generally symmetric IC.
    \end{itemize}
\end{corollary}

The above theorem and corollary show that any nontrivial representation must be induced by a family of informationally complete (quasi-)POVMs. It is not clear whether these representations must also be minimal. It could very well be that the proof of Theorem \ref{theor:trivialrep} can be adapted to show that any representation must necessarily be minimal to preserve functoriality. If this is the case then we automatically get that the object $n$ in $\CPM$ is mapped to $n^2$ in $\Qstoch$ providing a natural way of viewing the difference between the amount of perfectly distinguishable states ($n$) versus the dimension of the state space ($n^2$). This specific relation between these values is identified as something particular to quantum theory in for instance \cite{chiribella2011informational,appleby2011properties,hardy2001quantum}. Real-valued (or quaternion valued) quantum theory would have a different value here.

\bibliographystyle{eptcs}
\bibliography{main}

\appendix

\section{Proof of Theorem \ref{theor:nateq}} \label{app:nateq}
\begin{theorem*}
Let $\{E_i^{(n)}\}$ and $\{F_i^{(n)}\}$ be families of minimal IC-POVM's, and let the $T^1$ and $T^2$ be the respective families of matrices as defined in Theorem \ref{theor:embed}, and let $Q_i:\CPM \rightarrow \Qstoch_{T_i}$ be the respective functors. Then there exists a natural isomorphism $\eta: F_{T^1}\circ Q_1 \Rightarrow F_{T^2}\circ Q_2$.
\end{theorem*}
\begin{proof}
We need a family of transition matrices from the POVM $\{E_i^{(n)}\}$ to $\{F_i^{(n)}\}$. So first note that there is a matrix $\alpha^{(n)}$ such that
$$
\frac{E_j^{(n)}}{\tr(E_j^{(n)})} = \sum_k \alpha^{(n)}(k\lvert j) \frac{F_k^{(n)}}{\tr(F_k^{(n)})}.
$$
Then define the matrix $S^{(n)}$ by
$$
S^{(n)}(i\lvert j) = \tr(\frac{E_j^{(n)}}{\tr(E_j^{(n)})}F_i^{(n)}) = \sum_k \alpha^{(n)}(k\lvert j) \tr(\frac{F_k^{(n)}}{\tr(F_k^{(n)})} F_i^{(n)}) = \sum_k \alpha^{(n)}(k\lvert j) T^2_n(i\lvert k) 
$$
which means that $S^{(n)} = T^2_n \alpha^{(n)}$ or equivalently (using the fact that the POVM's are minimal so that their transition matrices are invertible) that $\alpha^{(n)} = \left(T^2_n\right)^{-1}S^{(n)}$.

Now given a CPTP map $\Phi: M_n\rightarrow M_m$ we can write
\begin{align*}
Q_1(\Phi)(j\lvert i) &= \tr(\Phi\left(\frac{E_i^{(n)}}{\tr(E_i^{(n)})}\right)E_j^{(m)}) = \tr(E_j^{(m)})\tr(\Phi\left(\frac{E_i^{(n)}}{\tr(E_i^{(n)})}\right)\frac{E_j^{(m)}}{\tr(E_j^{(m)})})\\
&= \tr(E_j^{(m)})\sum_{i^\prime,j^\prime} \alpha^{(n)}(i^\prime \lvert i)\alpha^{(m)}(j^\prime \lvert j)\frac{1}{\tr(F_{j^\prime}^{(m)})}\tr(\Phi\left(\frac{F_{i^\prime}^{(n)}}{\tr(F_{i^\prime}^{(n)})}\right)F_{j^\prime}^{(m)}) \\
&= \sum_{i^\prime, j^\prime} \alpha^{(n)}(i^\prime \lvert i)\alpha^{(m)}(j^\prime \lvert j)\frac{\tr(E_j^{(m)})}{\tr(F_{j^\prime}^{(m)})}Q_2(\Phi)(j^\prime \lvert i^\prime).
\end{align*}
This can be simplified by defining
$$
\beta^{(m)}(i\lvert j) := \alpha^{(m)}(j\lvert i) \frac{\tr(E_i^{(m)})}{\tr(F_j^{(m)})}.
$$
Note that the indexes of $\alpha$ and $\beta$ are reversed. This is intentional as this allows us to write
\begin{align*}
Q_1(\Phi)(j\lvert i) &= \sum_{j^\prime}\beta^{(m)}(j\lvert j^\prime) \sum_{i\prime} Q_2(\Phi)(j^\prime\lvert i^\prime)\alpha^{(n)}(i^\prime\lvert i) = \sum_{j^\prime}\beta^{(m)}(j\lvert j^\prime)\left(Q_2(\Phi)\alpha^{(n)}\right)(j^\prime\lvert i) \\
&= \left(\beta^{(m)}Q_2(\Phi)\alpha^{(n)}\right)(j\lvert i).
\end{align*}
What is this $\beta$? If we take $\Phi=id_n$ we get $Q_i(\Phi) = T^i_n$. By using $\alpha^{(n)} = \left(T^2_n\right)^{-1}S^{(n)}$ we then get 
$$
T^1_n = Q_1(\Phi) = \beta^{(n)}Q_2(\Phi)\alpha^{(n)} = \beta^{(n)} T^2_n \left(T^2_n\right)^{-1} S^{(n)} = \beta^{(n)}S^{(n)}.
$$
$S^{(n)}$ has an inverse since it acts as a basis transformation between the two POVM's. This allows us to write $\beta^{(n)} = \left(S^{(n)}\right)^{-1}T^1_n$ which means we get the following expression:
$$
Q_1(\Phi) = \left(S^{(m)}\right)^{-1}T^1_m Q_2(\Phi) \left(T^2_n\right)^{-1}S^{(n)}.
$$
Now by multiplying this equation with $\left(T^1_n\right)^{-1}$ on the right and using $(F_i\circ Q_i)(\Phi) = Q_i(\Phi)\left(T^i_n\right)^{-1}$ this translates to
$$
(F_1\circ Q_1)(\Phi) = \left(S^{(m)}\right)^{-1}T^1_m(F_2\circ Q_2)(\Phi) S^{(n)}\left(T^1_n\right)^{-1}.
$$
Define $\eta_n := S^{(n)}\left(T^1_n\right)^{-1}$. The equation can now be recast as
$$
\eta_m (F_1\circ Q_1)(\Phi) = (F_2\circ Q_2)(\Phi) \eta_n
$$
which means the collection of $\eta$'s forms a natural transformation. For all $n$, $S^{(n)}$ and $T^{(n)}$ are both invertible, so that $\eta^{(n)}$ is as well. Hence $\eta$ is a natural isomorphism.
\end{proof}

\section{Proof of Theorem \ref{theor:monoid}} \label{app:monoid}
\begin{theorem*}
    The functor $Q: \CPM\rightarrow \Qstoch_T$ as defined in Theorem \ref{theor:embed} for a family of minimal IC-POVM's is strong monoidal with coherence isomorphisms $\alpha_{n_1,n_2}=S_{n_1,n_2}$ where
    $$S_{n_1,n_2}(j\lvert i_1i_2) = \tr(\frac{E_{i_1}^{n_1}}{\tr(E_{i_1}^{n_1})}\otimes \frac{E_{i_2}^{n_2}}{\tr(E_{i_2}^{n_2})} E_j^{n_1n_2})
    $$
\end{theorem*}
\begin{proof}
    Let $\rho_i$ be states in $M_{n_i}$, then we have $\rho_i = \sum_j(T_{n_i}^{-1}Q(\rho_i))_j \frac{E_j^{n_i}}{\tr(E_j^{n_i})}$ which means we can write
    \begin{align*}
        Q(\rho_1\otimes \rho_2)_j &= \tr((\rho_1\otimes\rho_2)E_j^{n_1n_2}) = \sum_{i_1,i_2} (T_{n_1}^{-1}Q(\rho_1))_{i_1}(T_{n_2}^{-1}Q(\rho_2))_{i_2}\tr(\frac{E_{i_1}^{n_1}}{\tr(E_{i_1}^{n_1})}\otimes \frac{E_{i_2}^{n_2}}{\tr(E_{i_2}^{n_2})} E_j^{n_1n_2}) \\
        &= \sum_{i_1,i_2} S_{n_1,n_2}(j\lvert i_1, i_2)(T_{n_1}^{-1}Q(\rho_1))_{i_1}(T_{n_2}^{-1}Q(\rho_2))_{i_2}.
    \end{align*}
    This can now be written as $Q(\rho_1\otimes \rho_2) = S_{n_1,n_2}(T_{n_1}^{-1}\otimes T_{n_2}^{-1})(Q(\rho_1)\otimes Q(\rho_2))$ and substituting the definition of the tensor product of $\Qstoch_T$ $\otimes^\prime$ it becomes
    $$
        Q(\rho_1\otimes \rho_2) = S_{n_1,n_2}T_{n_1n_2}^{-1} (Q(\rho_1)\otimes^\prime Q(\rho_2)) = S_{n_1,n_2}*(Q(\rho_1)\otimes^\prime Q(\rho_2)).
    $$
    Now fix CPTP maps $\Phi_i: M_{n_i}\rightarrow M_{m_i}$ and write
    \begin{align*}
        Q(\Phi_1\otimes \Phi_2)*Q(\rho_1\otimes \rho_2) &= Q(\Phi_1(\rho_1)\otimes \Phi_2(\rho_2)) = S_{m_1,m_2}*(Q(\Phi_1(\rho_1))\otimes^\prime Q(\Phi_2(\rho_2))) \\
        &= S_{m_1,m_2}*(Q(\Phi_1)\otimes^\prime Q(\Phi_2))*(Q(\rho_1)\otimes^\prime Q(\rho_2)) \\
        &= S_{m_1,m_2}*(Q(\Phi_1)\otimes^\prime Q(\Phi_2))S_{n_1,n_2}^{-1}T_{n_1n_2}*Q(\rho_1 \otimes \rho_2).
    \end{align*}
    Because this has to hold for all states $\rho_i$ we can drop that term and we get the equality $Q(\Phi_1\otimes\Phi_2) = S_{m_1,m_2}*(Q(\Phi_1)\otimes^\prime Q(\Phi_2))S_{n_1,n_2}^{-1}T_{n_1n_2}$. Now by bringing some terms to the other side and using the definition of $*$ again we get the desired naturality equation:
    $$
        S_{m_1,m_2}*(Q(\Phi_1)\otimes^\prime Q(\Phi_2)) = Q(\Phi_1\otimes\Phi_2)*S_{n_1,n_2}.
    $$
    We still need to check that the coherence condition holds:
    \begin{align*}
        &S_{n_1n_2,n_3}*(S_{n_1,n_2}\otimes^\prime id_{n_3}) = S_{n_1,n_2n_3}*(id_{n_1}\otimes^\prime S_{n_2,n_3}) \iff \\
        &S_{n_1n_2,n_3}T_{n_1n_2n_3}^{-1}T_{n_1n_2n_3}(T_{n_1n_2}^{-1}\otimes T_{n_3}^{-1})(S_{n_1,n_2}\otimes T_{n_3}) = S_{n_1,n_2n_3}T_{n_1n_2n_3}^{-1}T_{n_1n_2n_3}(T_{n_1}^{-1}\otimes T_{n_2n_3}^{-1})(T_{n_1}\otimes S_{n_2,n_3}) \iff \\
        &S_{n_1n_2,n_3}(T_{n_1n_2}^{-1}S_{n_1,n_2}\otimes I_{n_3}) = S_{n_1,n_2n_3}(I_{n_1}\otimes T_{n_2n_3}^{-1}S_{n_2,n_3}).
    \end{align*}
    To prove this equality we need to write out $S_{n_1n_2,n_3}$ and to do that we first need the following: define $\beta$ such that
    $$
        \frac{E_{i_1}^{n_1n_2}}{\tr(E_{i_1}^{n_1n_2})} = \sum_{k_1,k_2} \beta(k_1,k_2\lvert i_1) \frac{E_{k_1}^{n_1}}{\tr(E_{k_1}^{n_1})}\otimes \frac{E_{k_2}^{n_2}}{\tr(E_{k_2}^{n_2})}
    $$
    which then gives
    \begin{align*}
        T_{n_1n_2}(l\lvert i_1) &= \tr(\frac{E_{i_1}^{n_1n_2}}{\tr(E_{i_1}^{n_1n_2})} E_l^{n_1n_2}) = \sum_{k_1,k_2} \beta(k_1,k_2\lvert i_1) \tr(\frac{E_{k_1}^{n_1}}{\tr(E_{k_1}^{n_1})}\otimes \frac{E_{k_2}^{n_2}}{\tr(E_{k_2}^{n_2})} E_l^{n_1n_2}) \\
        &= \sum_{k_1k_2}\beta(k_1,k_2\lvert i_1)S_{n_1,n_2}(l\lvert k_1,k_2) = (S_{n_1,n_2}\beta)(l\lvert i_1),
    \end{align*}
    so that $\beta = S_{n_1,n_2}^{-1}T_{n_1n_2}$.
    Now we can write
    \begin{align*}
        S_{n_1n_2,n_3}(j\lvert i_1, i_2) &= \tr(\frac{E_{i_1}^{n_1n_2}}{\tr(E_{i_1}^{n_1n_2})}\otimes \frac{E_{i_2}^{n_3}}{\tr(E_{i_2}^{n_3})} E_j^{n_1n_2n_3}) \\
        &= \sum_{k_1,k_2}\beta(k_1,k_2\lvert i_1) \tr(\frac{E_{k_1}^{n_1}}{\tr(E_{k_1}^{n_1})}\otimes \frac{E_{k_2}^{n_2}}{\tr(E_{k_2}^{n_2})}\otimes\frac{E_{i_2}^{n_3}}{\tr(E_{i_2}^{n_3})} E_j^{n_1n_2n_3}).
    \end{align*}
    Defining the quantity in the trace to be $S_{n_1,n_2,n_3}(j\lvert k_1,k_2,i_2)$ (note the comma's) this becomes
    $$
        S_{n_1n_2,n_3}(j\lvert i_1,i_2) = \sum_{k_1,k_2} S_{n_1,n_2,n_3}(j\lvert k_1,k_2,i_2)\beta(k_1,k_2\lvert i_1) = \sum_{k_1,k_2,k_3} S_{n_1,n_2,n_3}(j\lvert k_1,k_2,k_3)\beta(k_1,k_2\lvert i_1)\delta_{k_3,i_2}
    $$
    so that we get the equality $S_{n_1n_2,n_3} = S_{n_1,n_2,n_3}(\beta\otimes I_{n_3}) = S_{n_1,n_2,n_3}(S_{n_1,n_2}^{-1}T_{n_1n_2}\otimes I_{n_3})$.

    Filling this in the left-hand side of the coherence equation:
     $$S_{n_1n_2,n_3} (T_{n_1n_2}^{-1}S_{n_1,n_2}\otimes I_{n_3}) = S_{n_1,n_2,n_3} (S_{n_1,n_2}^{-1}T_{n_1n_2}\otimes I_{n_3})(T_{n_1n_2}^{-1}S_{n_1,n_2}\otimes I_{n_3}) = S_{n_1,n_2,n_3}.$$
    By doing a similar rewriting exercise for $S_{n_1,n_2n_3}$ we get the same expression on the right-hand side, which proves the coherence equation.
\end{proof}

\section{Proof of Theorem \ref{theor:trivialrep}}\label{app:trivialrep}

We first prove a short lemma about subspaces of matrices.
\begin{lemma*}
Let $L\subset M_n$ be a subspace containing the identity and at least one Hermitian matrix with at least two distinct eigenvalues, then $V=$ span$\left(\bigcup_{U\in U(n)} ULU^\dagger\right) = M_n$, where the union is taken over all unitary matrices.
\end{lemma*}
\begin{proof}
    We need to show that $V$ contains all matrices. Since it contains all unitary conjugations of any matrix, it suffices to show that it contains all diagonal matrices. Let $E$ be the matrix in $L$ with at least two distinct eigenvalues. We can diagonalise $E=UD_1U^\dagger$, so $D_1$ is in $V$. Write $D_1 = $ diag$(\lambda_1,\ldots,\lambda_n)$ taking $\lambda_1\neq \lambda_2$. Since $I_n$ is in $V$ we also have $D_2 = 1/(\lambda_2-\lambda_1)(D_1-\lambda_1 I_n)$ in $V$. We see that $D_2 = $ diag$(0,1,\lambda_3^\prime,\ldots,\lambda_n^\prime)$. Now we can apply the unitary conjugation to $D_2$ that interchanges the first and second coordinate and subtract it from $D_1$ giving $D_3 = D_2 - PD_2P = $ diag$(-1,1,0,\ldots,0)$. By considering conjugation with permutation matrices we can get the $-1$ and $1$ at arbitrary spots along the diagonal. The linear span of these operators is the set of diagonal matrices with zero trace. Because we also have the identity we can then create arbitrary diagonal matrices.
\end{proof}

\begin{theorem*}
    Let $F: \CPM \rightarrow \Qstoch$ be a quasi-stochastic representation of $\CPM$. This representation is either faithful or trivial.
\end{theorem*}
\begin{proof}
    Let $\{E_i^{(n)}\}$ denote the $n$th quasi-POVM associated to $F$.
    Let $V = $ span$\{E_i^{(n)}\}$ and let $V^\perp = \{A \in M_n ; \langle A, E_i^{(n)}\rangle_{HS} = 0 ~\forall i\}$. We note that $M_n = V \oplus V^\perp$. $F$ is faithful on $n$ if and only if $V^\perp = \{0\}$, since otherwise we can find a state $\rho = \rho_1 + \rho_2$ where $\rho_1\in V$, $\rho_2\in V^\perp$, $\rho_2\neq 0$ such that $F(\rho) = F(\rho_1)$.

    Because of functoriality we must have $F(\Phi(\rho)) = F(\Phi)F(\rho) = F(\Phi)F(\rho_1) = F(\Phi(\rho_1))$ for all CPTP maps $\Phi$. If we can find a $\Phi$ such that $\Phi(\rho_2)\not\in V^\perp$ then this leads to a contradiction because we would have $F(\Phi(\rho))\neq F(\Phi(\rho_1))$. Therefore for the POVM to fit in a valid representation we must have $V^\perp$ closed under application of an arbitrary CPTP map. In particular, it has to be closed under unitary conjugation: when $\tr(\rho_2 E_i^{(n)}) = 0$ for all $i$ we must also have $\tr(U\rho_2 U^\dagger E_i^{(n)}) = \tr(\rho_2 U^\dagger E_i^{(n)} U) = 0$ for all unitaries $U$. This means we should have $\tr(\rho_2 A) = 0$ for all $A\in W = $ span$\left(\bigcup_{U\in U(n)} UV U^\dagger\right)$.

    Now we distinguish two cases. Either all the $E_i^{(n)}$ are multiples of the identity in which case $F(\rho)_i = \tr(\rho E_i^{(n)}) = \tr(\rho \alpha_i I_n) = \alpha_i$ where $E_i^{(n)} = \alpha_i I_n$, so that the representation is trivial, or there is a $E_i^{(n)}$ that isn't a multiple of the identity in which case it has at least two distinct eigenvalues. In the latter case the space $W$ satisfies the conditions for the previous lemma which gives $W=M_n$. But then $\tr(\rho_2 A) = 0$ for all $A\in M_n$ so that necessarily $\rho_2 = 0$, which shows that $V^\perp = \{0\}$. The representation is then indeed faithful for $n$.

    Let us now suppose that $F$ isn't faithful. Then there is an $n$ such that two states in $M_n$ are mapped to the same distribution. We then know that all states are mapped to the same distribution for this $n$ by the argument above. Let $m\geq n$. Pick two orthogonal pure states $\ket{v}\bra{v}, \ket{w}\bra{w} \in M_m$ and let $\rho$ and $\sigma$ be orthogonal pure states in $M_n$, then there exists a partial isometry $\Phi$ such that $\Phi(\rho) = \ket{v}\bra{v}$ and $\Phi(\sigma) = \ket{w}\bra{w}$. By functoriality we get $F(\ket{v}\bra{v}) = F(\Phi(\rho)) = F(\Phi)F(\rho)=F(\Phi)F(\sigma) = F(\Phi(\sigma)) = F(\ket{w}\bra{w})$. Since $v$ and $w$ were arbitrary, all pure states must be mapped to the same distribution and by convexity this holds for all states. This means that the representation is trivial for all $m\geq n$. If we now consider a CPTP surjection $\Phi: M_n\rightarrow M_2$ we can use the same argument to show that the representation is trivial for $m=2$ which shows that indeed the entire representation is trivial.
\end{proof}

\end{document}